\documentclass[5p, twocolumn]{elsarticle}

\usepackage{amsthm}
\usepackage{amssymb}
\usepackage{amsmath}
\usepackage[latin1]{inputenc}
\usepackage{url}

\newtheorem{prop}{Proposition}[section]
\newtheorem*{nprop}{Proposition}
\newtheorem{defn}{Definition}[section]
\newtheorem{rmq}{Remark}[section]
\newcommand{\RR}{\mathbb{R}}
\newcommand{\PP}{\mathbb{P}}
\newcommand{\QQ}{\mathbb{Q}}
\newcommand{\hphi}{\hat \phi}
\newcommand{\hl}{\hat \lambda}
\newcommand{\hps}{\hat \psi}
\newcommand{\hyb}{\mathcal{H}yb}
\newcommand{\dhyb}{\mathrm{D}\mathcal{H}yb_{12}}
\newcommand\calM{\mathcal{M}}
\newcommand\calP{\mathcal{P}}

\newcommand\calC{\mathcal{C}}

\newcommand*{\QEDB}{\hfill\ensuremath{\square}}%
\def\*#1{\mathbf{#1}}

\begin{document}

\begin{frontmatter}



\title{Low-order continuous finite element spaces on hybrid non-conforming
    hexahedral-tetrahedral meshes}

\author{Maxence Reberol}
\author{Bruno L\'evy}
\address{INRIA Nancy Grand-Est}


\begin{abstract}
    This article deals with solving partial differential equations with the
    finite element method on hybrid non-conforming hexahedral-tetrahedral
    meshes. 
    By non-conforming, we mean that a quadrangular face of a hexahedron can be
    connected to two triangular faces of tetrahedra.
    We introduce a set of low-order continuous ($C^0$) finite element spaces
    defined on these meshes.
    They are built from standard tri-linear and quadratic Lagrange finite
    elements with an extra set of constraints at non-conforming
    hexahedra-tetrahedra junctions to recover continuity.
    We consider both the continuity of the geometry and the continuity of the
    function basis as follows: 
    the continuity of the geometry is achieved by using quadratic mappings for
    tetrahedra connected to tri-affine hexahedra and the continuity of
    interpolating functions is enforced in a similar manner by using quadratic
    Lagrange basis on tetrahedra with constraints at non-conforming junctions
    to match tri-linear hexahedra. 
    The so-defined function spaces are validated numerically on simple Poisson
    and linear elasticity problems for which an analytical solution is known. 
    We observe that using a hybrid mesh with the proposed function spaces
    results in an accuracy significantly better than when using linear
    tetrahedra and slightly worse than when solely using tri-linear hexahedra.
    As a consequence, the proposed function spaces may be a promising
    alternative for complex geometries that are out of reach of existing full
    hexahedral meshing methods.
\end{abstract}


\end{frontmatter}
\section{Introduction and related work}

In finite element methods, it is widely known that hexahedron finite elements
achieve better execution-time than tetrahedra ones for reaching a given accuracy.
Automatic tetrahedral meshing techniques are now mature and work well on any
complex 3D model \cite{frey2010}, \cite{si2015tetgen}. 
On the contrary, hexahedral meshing is still an open and difficult problem for
which there is still no satisfactory solution \cite{staten2007}.
Difficulties in hexahedral meshing can be partially resolved by introducing
other elements such as tetrahedra, pyramids and prisms, thus generating 
hybrid meshes.
Recent progress in hexahedral-dominant meshing techniques such as
\cite{baudouin2014}, \cite{botella2015}, \cite{sokolov2015} and \cite{bernard2016}
make it possible to automatically produce hybrid meshes with a large majority
of hexahedra for arbitrary 3D models.

Recently in the context of discontinuous Galerkin methods, hybrid meshes have
been successfully used on acoustic wave equation problems \cite{chan2015} or on
Maxwell equations \cite{bergot2013} with significant speedups over tetrahedral
meshes. For standard Galerkin methods, continuous finite element spaces for
hybrid meshes have been introduced, such as in \cite{sherwin1998}. A succinct
survey exposing various approaches with emphasis on the pyramidal element is
available in the introduction of \cite{bergot2010}. All these propositions
involve special functions, such as rationals, as it is not possible to build a
polynomial function basis on the pyramid which is conforming with tetrahedra
and hexahedra polynomial function basis, as noticed in \cite{bedrosian1992}.

In the present article, we adopt a different approach in which we consider
hybrid meshes composed only of hexahedra and tetrahedra, for which finite element
behavior is very well understood. A second interesting point of this approach
is that enabling non-conforming hexahedra-tetrahedra junctions provides more
flexibility to hex-dominant meshing techniques, resulting in a higher
proportion of hexahedra, as it eliminates constraints associated with the
generation of pyramids. However, without special care, non-conformities in the
mesh result in a discontinuous geometry and a discontinuous function space. 

The idea of using non-conforming hexahedral-tetrahedral meshes is not new and
has been successfully developed in the context of the discontinuous Galerkin
Method in electromagnetic in \cite{durochat2013}, \cite{leger2014} and
\cite{fahs2015}. 
For continuous Galerkin methods, constraints to ensure continuity of the
divergence and of the rotational along non-conforming interfaces have been
briefly proposed in \cite{marais2008}. 
In a engineering approach, hexahedra-tetrahedra non-conforming junction have
been firstly discussed in \cite{dewhirst1993} which proposes various multi-point
constraints to ensure the function continuity or to minimize the error,
depending of the finite element considered in their software. An extention of this
approach \cite{owen1997} discards non-conforming hexahedral-tetrahedral
junctions in favor of pyramidal elements.
Our contribution is to give a formal approach to this problem and to deal with
the geometric discontinuity arising with non-planar hexahedra faces, which was
not considered in previous work to our knowledge.

It should also be noted that previous works on finite element over hybrid
meshes usually consider applications where hexahedral and tetrahedral elements
lie in distinct regions. Transitional elements or non-conforming junctions
arises then in localized layers or regions. This approach is especially
efficient for problems such as acoustic where complex objects are meshed with
tetrahedra and the propagation medium with hexahedra. Our approach is more
oriented toward hex-dominant meshes where tetrahedra are located randomly in
the mesh, resulting in a high number of non-conforming junctions scattered 
randomly in the domain.

In the present article, we introduce low-order continuous function spaces
defined on hybrid non-conforming hexahedral-tetrahedral meshes. The geometric
conformity is obtained by using quadratic mappings for tetrahedra to exactly
fit the hexahedra non-planar faces. Likewise, quadratic Lagrange basis and
constraints are used on tetrahedra to produce functions which are continuous
($C^0$) at interfaces with the tri-linear functions used in hexahedra.

\section{Continuous function spaces on hybrid hexahedral-tetrahedral meshes}
\label{sec:continuous}
Our first goal is to deliver a \emph{continuous} geometry for the mesh, in the sense
explicited below. We will then explain how to define a \emph{continuous} function space
on this geometry.

\paragraph*{Input} The input is a mesh composed of a set of vertices
(geometric information) and a set of elements defined by their vertices and
faces (combinatorial information). In the present article, we restrict
ourselves to meshes that satisfy the following specification:

\begin{defn}{\textit{Combinatorial hybrid hexahedral-tetrahedral mesh
        specification}} \\ 
    The input hybrid mesh $\calM$ is composed of a set $\calP$ of vertices,
    defined by their coordinates, a set of tetrahedra defined by their 4
    vertices in $\calP$, and a set of hexahedra defined by their 8 vertices in
    $\calP$ and their 6 faces (defined by 4 vertices in $\calP$). The
    connectivity is restricted to the following combinatorial cases:
    \begin{itemize}
       \item Two tetrahedra share 0, 1, 2 or 3 vertices.
       \item Two hexahedra share 0, 1, 2 or 4 vertices.  When they share 2
           vertices, this is a common edge.  When they share 4 vertices, this a
           common face.
       \item One hexahedron and one tetrahedron share 0, 1, 2 or 3 vertices.
           When they share 3 vertices, there exists another tetrahedron which
           also shares 3 vertices with the hexahedron and two or three vertices
           with the tetrahedron. So in this setup, the hexahedron face is
           connected to 2 tetrahedra faces. \QEDB
    \end{itemize}
    \label{defn:inputspec}
\end{defn}

This specification permits non-conforming connections between a hexahedron
and two tetrahedra, that we often refer as hybrid junction. But
it excludes all other types of non-conforming connections.  Examples of supported
and not supported configurations are shown in figure \ref{fig:good_bad_hjunction}.
\begin{figure}
  \centering
  \includegraphics[width=\columnwidth]{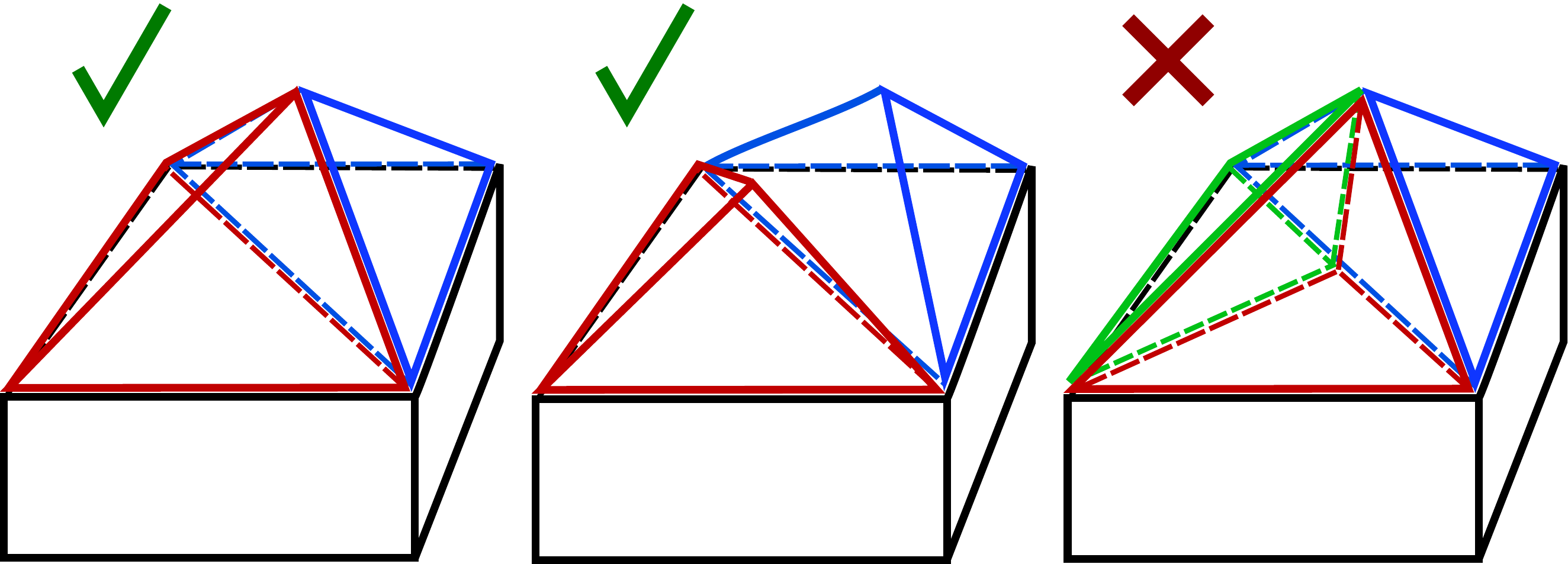} 
  \caption{Supported and not supported non-conforming junctions between
  hexahedra and tetrahedra}
  \label{fig:good_bad_hjunction}
\end{figure}

\subsection{Mesh geometry}
The input mesh, defined by its combinatorial information, does not provide a
geometry. A naive idea would be to use affine tetrahedra and tri-affine
hexahedra but this solution leads to gaps or overlaps between hexahedra and 
tetrahedra at non-conforming junctions when hexahedron faces are not planar
(see figure \ref{fig:geo_hjunction}a.). The mesh geometry we are looking for
should satisfy definition \ref{defn:hmesh}.

\begin{defn}\label{defn:hmesh}
    {\textit{Geometric hybrid hexahedral-tetrahedral mesh}} \\ 
    A hybrid hexahedral-tetrahedral mesh is the union of a set of (non-degenerate)
    hexahedra and of a set of (non-degenerate) tetrahedra. The cell geometries
    $K_c$ of $\Omega_h$ satisfy:
\begin{itemize}
  \item \(\Omega_h = \bigcup_{i=c}^N K_c\)
  \item the intersection \(K_i \cap K_j\) of two distinct tetrahedra is either
      empty, or reduced to a common vertex, or an entire common edge, or an
      entire common face (triangle)
  \item the intersection \(K_i \cap K_j\) of two distinct hexahedra is either
      empty, or reduced to a common vertex, or an entire common edge, or an
      entire common face (quadrilateral)
  \item the intersection \(K_i \cap K_j\) of an hexahedron and a tetrahedron is
      either empty, or a common vertex, or an entire common edge, or an entire
      quadrilateral diagonal, or a triangle such that there exists another
      tetrahedra $K_l$ which intersection with $K_i$, $K_i \cap K_l$, is
      another tetrahedron face and \(K_i \cap (K_j \cup K_l)\) is a
      quadrilateral face of $K_i$ \QEDB
\end{itemize}
\end{defn}

Following the standard finite element approach, we define cells as images of
the reference tetrahedron $\hat T$ and of the reference hexahedron $\hat Q$ 
(detailed in appendices) by one-to-one mapping functions $\*F_c$:
\begin{align*}
    K_c =\ &\*F_c(\hat T) \text{ if the cell } c \text{ is a tetrahedron} \\
    K_c =\ &\*F_c(\hat Q) \text{ if the cell } c \text{ is a hexahedron}
\end{align*}
The edges and the faces in the definition \ref{defn:hmesh} are images of edges and faces
of the reference hexahedron or tetrahedron (so they can be curved).

The first question we consider is how to define the mappings $\*F_c$, for both
tetrahedra and hexahedra, in order to satisfy the general definition
\ref{defn:hmesh} (geometric continuity).

\paragraph*{Hexahedron mappings}
Let us start by considering the mapping $\*F_Q$ of a hexahedron $K_q \in \Omega_h$.
The standard mapping from the reference hexahedron $\hat Q$ (unit cube) is the
so-called tri-affine mapping. It is based on the function space of
polynomials of degree one in each variable $\QQ_1$. Each component $F_{Q,j}$ is
in $\QQ_1$ so $\*F_{Q}$ is in $(\QQ_1)^3$. We have the decomposition:
$$ \forall \hat{\*p} \in \hat Q, \quad \*F_{Q} (\hat{\*p}) = \sum_{i=1}^8 \*a_i \ \hps_i(\hat{\*p}) $$ 
where $(\hps_i)_{i=1..8}$ is the basis of $\QQ_1$ detailed in the appendix \ref{app:cube}
and $\*a_i$ are the vertices of the hexahedron $Q$ given in the input mesh.


It is important to notice that this tri-affine mapping has components which are
polynomials of degree 3 (product of three degree one).  If we consider the
restriction to a face of $\hat Q$, then the restricted mapping is bi-affine,
i.e. a bi-variate polynomial of degree 2. 
It implies that the surface of the mapped face is a quadric, specifically a
hyperbolic paraboloid. So for an arbitrary hexahedron, its faces are not planar
in general.

\begin{figure}
  \centering
  \includegraphics[width=\columnwidth]{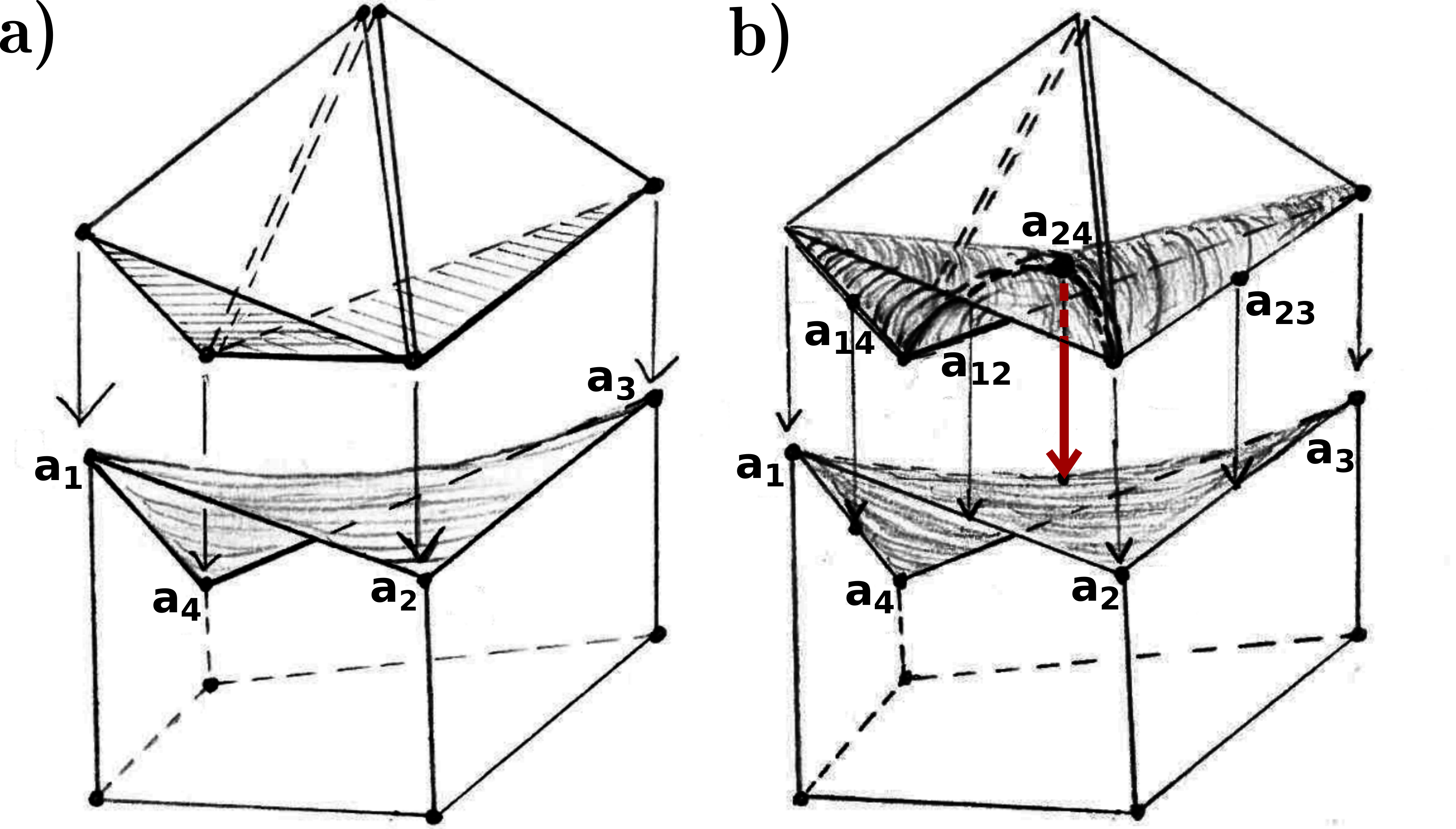} 
  \caption{Non-conforming hexahedron-tetrahedra junction. a) Tetrahedra affine 
      mappings, b) Tetrahedra quadratic mappings}
  \label{fig:geo_hjunction}
\end{figure}

This is not an issue for meshes composed only of hexahedra because they share
common vertices at element junctions. But for hybrid hexahedral-tetrahedral mesh,
it will not always be possible to glue two tetrahedra (planar faces) with a tri-affine
hexahedron, which faces are quadric surfaces. This incompatibility is
illustrated in the figure \ref{fig:geo_hjunction}a. 
Unfortunately, automated hex-dominant mesh generation algorithms, such as
\cite{sokolov2015}, produce hexahedra with non-planar faces most of the time as
they solely use combinatorial definitions. Therefore, we need to take these
particularities into account in our finite element mappings.

\paragraph*{Tetrahedron mappings}
The simplest mapping $\*F_T$ for an arbitrary tetrahedron is the affine
one, where the three components lie in the space of polynomials of degree
one $\PP_1$. However this mapping generates faces which are planar, so it would
not be possible in general to continuously connect a tetrahedron to an
arbitrary hexahedron (which faces are quadric).

We propose to solve this issue by using quadratic mappings for tetrahedra
(see figure \ref{fig:map_tet_P2}). It allows us to deform tetrahedra geometry
in order to fit exactly with the quadric hexahedron faces at hybrid interfaces.

\begin{figure}
  \centering
  \includegraphics[width=\columnwidth]{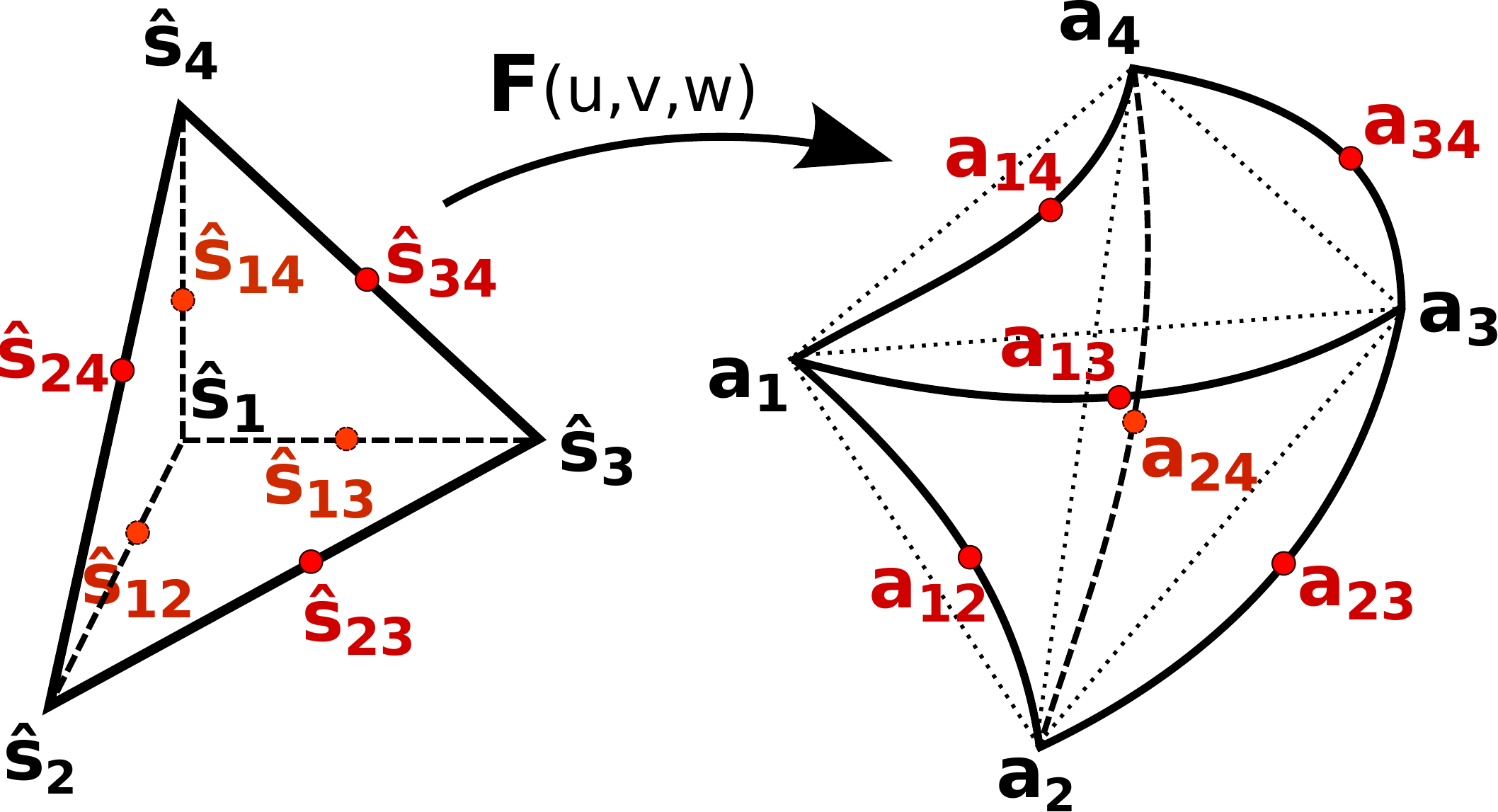} 
  \caption{Quadratic mapping of the reference tetrahedron}
  \label{fig:map_tet_P2}
\end{figure}

Consider the mapping $\*F_T$ which maps the reference tetrahedron $\hat T$ to
the actual tetrahedron $T$.  Instead of taking $\*F_T$ in $(\PP_1)^3$, we
take $\*F_T$ in $(\PP_2)^3$, where $\PP_2$ is the space of polynomials of
degree 2, detailed in the appendix \ref{app:tet}. We have the decomposition: 
$$ \*F_T = \sum_{i=1}^{4} \*a_i \ \hphi_i + \sum_{1 \leq i < j \leq 4} \*a_{ij} \ \hphi_{ij}$$
where the $\*a_i$'s are the vertices of the tetrahedron $T$ given in the input mesh,
the $\*a_{ij}$'s are the midpoints of edges $i-j$ and $(\hphi_i, \hphi_{ij})_{i, ij}$ is
the basis of $\PP_2$ detailed in the appendix \ref{app:tet}.

By moving the midpoints $\*a_{ij}$, it is possible to control the
deviation from the affine tetrahedron of the geometry, as shown in figure
\ref{fig:map_tet_P2}.
Initially, the $\*a_{ij}$ are set as $\*a_{ij} = \frac{\*a_i + \*a_j}{2}$,
this reproduces the affine mapping but we now have the freedom to change one of
these coefficients to deform the geometry of tetrahedra.

We assume that the $\*a_{ij}$ coefficients are chosen in such way  that $\*F_T$
remains a one-to-one mapping. This is the case if the deviation from the actual
edge midpoint is sufficiently small. For a more detailed discussion on the
quadratic geometry validity, one can refer to \cite{george2012}. This type of
curved geometry, often referred to as isoparametric elements, is used in the
finite element framework to produce meshes that better fit non-polygonal
boundaries of 3D models.

At this point, the mesh geometry is defined by:
\begin{align*}
    \Omega_h = \bigcup_c &K_c \ \text{ with } \\
    K_c = &\*F_c(\hat T),\ F_c \in (\PP_2)^3 \text{ if } K_c \text{ is a tetrahedron} \\
    K_c = &\*F_c(\hat Q),\ F_c \in (\QQ_1)^3 \text{ if } K_c \text{ is a hexahedron}
\end{align*}

To satisfy definition \ref{defn:hmesh}, we need to constrain the new
degrees of freedom in tetrahedra (edge midpoints) at hybrid interfaces to ensure
that the elements match exactly.

\paragraph*{Continuity of the geometry at interfaces between elements} $\ $ \\
\indent 1. For two connected hexahedra, both geometries match on the common face
because both tri-affine mappings are fully determined by their values at
reference vertices, and these values are nothing else than the hexahedron
vertices of the combinatorial definition of the input mesh.\\

2. At the interface between two connected tetrahedra defined by quadratic
mappings, the geometry continuity is achieved if both quadratic mappings share
6 coefficients, associated with values taken at vertices and edge midpoints
of faces of the reference tetrahedron. Three of these equalities, at vertices,
are guaranteed by the input mesh specification \ref{defn:inputspec}. The last
three, at edge midpoints of reference faces, are added as constraints to our
geometric mesh definition. This constraint is referred as $(\mathcal{C}\text{-T})$.\\

3. For hybrid hexahedron-tetrahedron interface, such as in figure
\ref{fig:geo_hjunction}, the key to achieve the continuity is to set the
coefficient which corresponds to the hexahedron face diagonal to the point at
the center of the face in both tetrahedra mappings. In figure
\ref{fig:geo_hjunction}b., $\*a_{24}$ is set to
$\frac{\*a_1+\*a_2+\*a_3+\*a_4}{4}$ instead of $\frac{\*a_2+\*a_4}{2}$. Other
tetrahedra mapping coefficients associated with edge midpoints, such as $a_{23}$,
are left at actual edge midpoints. The formalisation of this claim is given by the
proposition \ref{prop:hjunction}.\\

\textit{Notations}
To lighten formulas, we adopt the following convention:
the pre-images of vertices $\*a_i$ of the mesh by element mappings $\*F_e$
are denoted by $\hat{\*a}_i = \*F_e^{-1}(\*a_i)$. These
$\hat{\*a}_i$ are vertices of the reference elements. When more than one 
mapping is involved in a formula, the right mapping can be deduced from the
domain of the function which applies to $\hat{\*a}_i$. One should be careful
that it implies that in the same formula, two reference points denoted by
$\hat{\*a}_i$ can refer to two distinct points (but by renumbering of nodes in
the mappings, it is always possible to get to a configuration where both refer to
the same point).

\begin{prop}{Non-conforming hexahedron-tetrahedra junction}{\label{prop:hjunction}} \\
Let $Q$ be a hexahedron defined by the tri-affine mapping $\*F_Q$ of the
reference hexahedron such that $\*a_1, \*a_2, \*a_3, \*a_4$ are the vertices of the
face $\Sigma_q$ of $Q$. \\ Let $T$ be a tetrahedron defined by the quadratic mappings
$\*F_T$ of the reference tetrahedron such that $\*a_1, \*a_2, \*a_3$ are the
vertices of the face $\Sigma_t$ of $T$.

Then $\Sigma_t \subset \Sigma_q$, i.e. the geometry of $Q \cup T$ is continuous, if 

\begin{itemize}
  \item For common edges  \([\*a_i, \*a_j], \ (i,j) = (1,2), (1,4)\) \\
      $$ \*F_T(\frac{\hat{\*a}_i + \hat{\*a}_j}{2}) = \frac{\*a_i + \*a_j}{2} $$
  \item For the diagonal $(a_2, a_4)$ of the face $\Sigma_q$, which is also a
      edge of $\Sigma_t$: 
      $$ \*F_T(\frac{\hat{\*a}_2 + \hat{\*a}_4}{2}) =
      \frac{\*a_1 + \*a_2 + \*a_2 + \*a_4}{4}$$
\end{itemize}
\end{prop}
\begin{proof}
    See supplemental material.
\end{proof}

The final mesh geometry with an explicit definition of the mappings is given by 
the following definition.

\begin{defn} Let $\calM$ be the input combinatorial hybrid mesh which satisfies
    the specifications \ref{defn:inputspec}. For each element of $\calM$,
    vertices are denoted by $\*a_j$ where $j$ is the local index. $\hat{\*a}_j$ is
    the pre-image by the mapping $\*F_i$ of the associated element.  
    
    The space partition of $\Omega$ is the union $\Omega_h = \bigcup_{i=1}^N
    K_i$. The $N$ element geometries are defined by the following mappings
    $(\*F_i)_{i=1..N}$:

\begin{itemize}
    \item If the $i$-th element is a hexahedron,  
        $K_i = \*F_i(\hat Q)$ with $\*F_i \in (\QQ_1)^3$ determined by
        $\*F_i(\hat{\*a}_j) = \*a_j, \ 1 \leq j \leq 8$
    \item If the $i$-th element is a tetrahedron, 
        $K_i = \*F_i(\hat T)$ with $\*F_i \in (\PP_2)^3$ determined by
        $\*F_i(\hat{\*a}_j) = \*a_j, \ 1 \leq j \leq 4$ and for $ 1 \leq j < k
        \leq 4$:
        \begin{itemize}
            \item $\*F_i(\frac{\hat{\*a}_{j}+\hat{\*a}_k}{2}) = \frac{\*a_j +
                    \*a_k + \*a_l + \*a_f}{4}$ if $(\*a_j, \*a_k)$ is a
                diagonal of a hexahedron face, whose vertices are $\*a_j, \*a_l,
                \*a_k, \*a_f$.
            \item $\*F_i(\frac{\hat{\*a}_{j}+\hat{\*a}_k}{2}) = \frac{\*a_j +
                    \*a_k}{2}$ else
        \end{itemize}
\end{itemize}
\label{defn:geomesh}
\end{defn}

This definition \ref{defn:geomesh} satisfies both constraints
$(\mathcal{C}\text{-T})$ at tetrahedra interfaces and the assumptions of the
proposition \ref{prop:hjunction} at hybrid interfaces.
So a mesh geometry defined by \ref{defn:geomesh} satisfies the definition
\ref{defn:hmesh} (continuity of the geometry).

\subsection{Continuity of the function spaces}

From now on, we assume that $\Omega_h$ is a partition of the domain $\Omega$
given by the definition \ref{defn:geomesh}.  The objective is to use the hybrid
mesh $\Omega_h$ to build a function space, in which functions are piecewise
defined (element by element). Following the standard finite element approach of Ciarlet
\cite{ciarlet1978}, we use two ingredients: function spaces described on
reference elements and one-to-one element mappings. More specifically:

\noindent
1. For our hybrid function space, we use the low-order polynomial spaces
$\QQ_1$ and $\PP_2$ which have simple Lagrange-based basis $(\hps_i)_{i=1..8}$,
$(\hphi_i)_{i=1..10}$, defined respectively on the reference hexahedron $\hat
Q$ and the reference tetrahedron $\hat T$ (see appendix).

\noindent
2. The inverses of the element mappings defined in \ref{defn:geomesh} are used to
get to the reference elements from anywhere in the actual mesh $\Omega_h$:
\begin{align*}
\forall \*p \in \Omega_h, \quad &
&&\exists \*F_Q \in (\QQ_1)^3 \text{ such that } \hat{\*p}= \*F_Q^{-1}(\*p) \in
\hat Q \\
&\text{or}&&\exists \*F_T \in (\PP_2)^3 \text{ such that } \hat{\*p} =
\*F_T^{-1}(\*p) \in \hat T
\end{align*}
Note that in practice, the inverse mappings never need to be computed explicitly
(see section \ref{sec:apppoisson}). \\

Consider a hexahedron $K_q \in \Omega_h$ and a function $\hat f \in \QQ_1$
defined by its values $(f_i)_{i=1..8}$ at the 8 vertices $\hat q_i$ of $\hat
Q$. Then by composition, we form the function $f_{|K_q}$ defined by:
$$ \forall \*p \in K_q, \; f_{|K_q}(\*p) = \hat f \circ F_Q^{-1}(\*p) 
    =  \hat f (\hat{\*p}) $$ 
The same construction can be done for a tetrahedron $K_t \in \Omega_h$ by
taking $\hat f \in \PP_2$. So our space is composed of functions, whose
restrictions on each elements are defined by the composition of a polynomial
(in $\QQ_1$ or in $\PP_2$) and the inverse mapping of the element. We name it
$\dhyb$, for discontinuous hybrid space:
\begin{align*}
    \dhyb = \{ f \in& L^2(\Omega_h) \text{ such that } f_{|K_c} = \hat f
        \circ F_{K_c}^{-1}, \\
    &\hat f \in \QQ_1 \text{ if } K_c \in \Omega_h \text{ is a hexahedron} \\ 
    &\hat f \in \PP_2 \text{ if } K_c \in \Omega_h \text{ is a tetrahedron} \}
\end{align*}

On each element, the functions of this space are continuous (composition of a
polynomial and a continuous mapping). The continuity at element interfaces
needs to be enforced (see below).

\begin{rmq} 
    It is important to notice that the space $\dhyb$ is not composed of
    polynomials, because if mapping components are polynomials (of degree 2),
    inverse mappings are not.
\end{rmq}

As stated in the section introduction, our goal is to build a \emph{continuous}
function space $\hyb \subset \calC^0(\Omega)$ which is suitable for classic
finite element methods. We achieve this by adding constraints at interfaces
between elements in $\dhyb$. We propose two continuous function spaces:
$\hyb_{12}$ and $\hyb_1$. $\hyb_{12}$ is the space with the minimum of
constraints applied to $\dhyb$ to ensure continuity and $\hyb_1$ is a space
with more constraints but easier to manipulate, more in the spirit of our
initial objectives.

Let us look at interfaces between elements to determine explicitly the
constraints. We only consider surface interfaces because continuity at edges is
guaranteed by continuity at surface interfaces.

\paragraph*{Continuity conditions at element surface interfaces} $\ $\\
\indent 1. \textit{Between two hexahedra.} A function $f \in \dhyb$ is continuous
at a hexahedra interface if its restrictions to both elements are equal at the common face vertices.
Below is a more detailed explanation: \\
Let consider two connected hexahedra $Q_1, Q_2$
which share the face $\Sigma_q$, with vertices $\*a_1, \*a_2, \*a_3,
\*a_4$. The tri-affine mappings are respectively $\*F_1$ and $\*F_2$.
The pre-images of the face by the mappings are denoted by $\hat \Sigma_1 =
\*F_1^{-1}(\Sigma_q),\ \hat \Sigma_2 = \*F_2^{-1}(\Sigma_q)$.  \\ 
The restrictions to the common face $f_{|Q_1 \cap \Sigma_q}, f_{|Q_2 \cap
    \Sigma_q}$ can be decomposed as: 
$f_{|Q_1 \cap \Sigma_q} = \hat f_{1|\hat \Sigma_1} \circ \*F_{1|\Sigma_q}^{-1}$
and $f_{|Q_2 \cap \Sigma_q} = \hat f_{2|\hat \Sigma_2} \circ
\*F_{1|\Sigma_q}^{-1}$.  Both $\hat f_{1|\Sigma_1}, \hat f_{2|\Sigma_2}$ are
determined by their values at vertices of $\hat \Sigma_1, \hat \Sigma_2$ (see
appendix). So both restrictions are equal if 
\begin{align*}
\forall i \in [1,4], \quad \quad 
&\hat f_{1|\hat \Sigma_1}(\*F_1^{-1}(\*a_i)) = \hat f_{2|\hat
    \Sigma_2}(\*F_2^{-1}(\*a_i)) \\ \Leftrightarrow \quad &f_{|Q_1}(\*a_i) =
f_{|Q_2}(\*a_i)
\end{align*}

We denote $(\calC\text{-}\mathrm{I})$ this continuity condition at
hexahedra interfaces
Using the notation convention introduced for proposition \ref{prop:hjunction},
this can be re-written as:
$$ \forall i \in
[1,4], \quad \hat f_{1|\hat \Sigma_1}(\hat{\*a}_i) = \hat f_{2|\hat
    \Sigma_2}(\hat{\*a}_i) $$ \\

2. \textit{Between two tetrahedra.} A function $f \in \dhyb$ is continuous
at a tetrahedra interface if its values are equal at the three vertices and
at the three edge midpoints of the common face. More specifically:

Let consider two connected tetrahedra $T_1, T_2$ which share the face $\Sigma_q$,
with vertices $\*a_1, \*a_2, \*a_3$.  The quadratic mappings are
respectively $\*F_1$ and $\*F_2$.  The restrictions $f_{|T_1 \cap \Sigma_q},
f_{|T_2 \cap \Sigma_q}$ are determined by their values at the three vertices and at
the three edge midpoints of the pre-images of $\Sigma_q$ by the mappings (see
appendix), denoted by $\hat \Sigma_1 = \*F_1^{-1}(\Sigma_q),\ \hat
\Sigma_2 = \*F_2^{-1}(\Sigma_q)$. Both restrictions are equal if:

\begin{itemize}
  \item For common vertices $a_i, \ i \in {1, 2, 3}$:
$$\hat f_{1|\hat \Sigma_1}(\hat{\*a}_i) = \hat f_{2|\hat \Sigma_2}(\hat{\*a}_i) $$

  \item For common edges \([a_i, a_j], \ (i,j) = (1,2), (1,3), (2,3)\): 
$$\hat f_{1|\hat \Sigma_1}(\frac{\hat{\*a}_i+\hat{\*a}_j}{2}) = \hat f_{2|\hat \Sigma_2}(\frac{\hat{\*a}_i+\hat{\*a}_j}{2}) $$
\end{itemize}

We denote $(\calC\text{-}\mathrm{II})$ this continuity condition at
tetrahedra interfaces. \\

3. \textit{Between one hexahedron and one tetrahedron.} A function $f \in
\dhyb$ is continuous at a non-conforming hexahedron-tetrahedron interface if its
values are equal at the three common vertices, at the two common edge midpoints
and at the quadrilateral center, which is a edge midpoint of the triangle. We
formalize this last continuity condition $(\calC\text{-}\mathrm{III})$ with
proposition \ref{prop:hjunction_c0}:

\begin{prop}\label{prop:hjunction_c0}{Continuity of the function spaces at hybrid junctions} \\
Let $Q$ be a hexahedron and $T$ be a tetrahedron which share the triangular
face $\Sigma_t$ (vertices $\*a_1, \*a_2, \*a_4$) of $T$. The associated
quadrilateral face of $Q$ is denoted by $\Sigma_q$ (vertices $\*a_1, \*a_2, \*a_3,
\*a_4$). This configuration is shown in figure \ref{fig:geo_hjunction}b.. The
element mappings $\*F_Q, \*F_T$ satisfy the mesh definition \ref{defn:geomesh}.
\\ Let $f_h \in \dhyb$. Its restrictions  $f_{|Q}, f_{|T}$ are defined by the
compositions $f_{|Q} = \hat f_Q \circ \*F_Q^{-1}$ and $f_{|T} = \hat f_T \circ
\*F_T^{-1}$. \\ $f_h$ is continuous at the hybrid interface $\Sigma_t$, i.e.
$f_{|T \cap \Sigma_t} = f_{|Q \cap \Sigma_q} \text{ on } \Sigma_t$, if:

\begin{enumerate}
  \item At common vertices $\*a_i, \ i \in {1, 2, 4}$:
$$\hat f_{T}(\hat{\*a}_i) = \hat f_{Q}(\hat{\*a}_i) $$
      (these points are degree of freedom for both functions $\hat f_Q \in
      \QQ_1, \hat f_T \in \PP_2$)

  \item At common edges \((\*a_i, \*a_j), \ (i,j) = (1,2), (1,4) \):
$$
\hat f_{T}(\frac{\hat{\*a}_i + \hat{\*a}_j}{2}) =
\hat f_{Q}(\frac{\hat{\*a}_i + \hat{\*a}_j}{2})
$$
(these points are degree of freedom only for $\hat f_T \in \PP_2$)

  \item At the diagonal $(\*a_2, \*a_4)$ of $\Sigma_q$:
$$
\hat f_{T}(\frac{\hat{\*a}_2 + \hat{\*a}_4}{2}) =
\hat f_{Q}(\frac{\hat{\*a}_1+\hat{\*a}_2+\hat{\*a}_3+\hat{\*a}_4}{4})
$$
      (this point is a degree of freedom only for $\hat f_T \in \PP_2$)
\end{enumerate}
\end{prop}
\begin{proof} See supplemental material.
\end{proof}

By solely considering functions $f \in \dhyb$ which satisfy conditions
$(\calC\text{-}\mathrm{I})$, $(\calC\text{-}\mathrm{II})$,
$(\calC\text{-}\mathrm{III})$,  we form the function space $\hyb_{12}$ defined
as follow:
\begin{defn}{\textit{Continuous hybrid hexahedral-tetrahedral function space}} \\
Let $\Omega_h$ a partition of $\Omega$ which satisfies definition
\ref{defn:geomesh}.
\begin{align*}
   \hyb_{12} = \{ v \in \calC^0(\Omega_h) \text{ such that } v_{|K_i} = \hat v
   \circ F_{K_i}^{-1} \\
                \hat v \in \PP_2 \text{ if } K_i \text{ is a tetrahedron} \\
                \hat v \in \QQ_1 \text{ if } K_i \text{ is a hexahedron}
           \} 
\end{align*}
\label{defn:hyb12}
\end{defn}
It should be stressed that the space $\hyb_{12}$ is composed of functions formed from
$\QQ_1$ and $\PP_2$. But as in practice the proportion of tetrahedra is low, most
tetrahedra are connected to hexahedra and their degree of freedom on edge midpoints
are constrained and do not contribute to the solution approximation. Considering
this remark, it can be interesting to also remove the remaining edge midpoint degrees
of freedom at tetrahedra interfaces. This produces a smaller
continuous function space that we call $\hyb_1$. In practice it is achieved by
changing the continuity condition $(\calC\text{-}\mathrm{II})$ to a more
constraining one, denoted by $(\calC\text{-}\mathrm{II}\text{-b})$, which
forces function values at tetrahedra edge midpoints to be the average of function
values at edge vertices (when these edges are not hexahedron face diagonal). The
second point of $(\calC\text{-}\mathrm{II})$ becomes:
\begin{itemize}
  \item At common edges \([a_i, a_j]\) of tetrahedra interfaces which are not
      hexahedron face diagonal: 
\begin{align*}
    &\hat f_{1|\hat \Sigma_1}(\frac{\hat{\*a}_i+\hat{\*a}_j}{2}) = \hat f_{2|\hat \Sigma_2}(\frac{\hat{\*a}_i+\hat{\*a}_j}{2}) \\
    = &\frac{\hat f_{1|\hat \Sigma_1}(\hat{\*a}_i) + \hat f_{1|\hat \Sigma_1}(\hat{\*a}_j)}{2}
     = \frac{\hat f_{2|\hat \Sigma_2}(\hat{\*a}_i) + \hat f_{2|\hat \Sigma_2}(\hat{\*a}_j)}{2}
  \end{align*}
\end{itemize}
With this condition $(\calC\text{-}\mathrm{II}\text{-b})$, the resulting
function space $\hyb_1$ is still continuous but function restrictions to
tetrahedra which are not connected to hexahedra are formed from $\PP_1$
(polynomials of degree 1). More importantly, functions in $\hyb_1$ are entirely
defined by their values at mesh vertices, resulting in functions easier to
manipulate and a smaller linear system in the finite element method.

\begin{defn}{\textit{Minimal continuous hybrid hexahedral-tetrahedral function space}} \\
Let $\Omega_h$ a partition of $\Omega$ which satisfies definition
\ref{defn:geomesh}.
\begin{align*}
    \hyb_{1} = \{ v \in \hyb_{12} \text{ that satisfy } (\calC\text{-}\mathrm{II}\text{-b})\}
\end{align*}
\label{defn:hyb1}
\end{defn}
The introduced function spaces satisfy the following inclusions:
$$ \hyb_1 \subset \hyb_{12} \subset \dhyb$$
$$ \hyb_1 \subset \hyb_{12} \subset \calC^0$$

\subsection{Function space basis}
\label{sec:basis}
%

This section details explicitly the function basis of the space $\hyb_1$ and
$\hyb_{12}$.  One difficulty is these spaces have been built by using
constraints that depend of the mesh local combinatorial configuration (element
type, hybrid interface or not) and that cannot be applied blindly in a generic
way. So it is not straightforward to expose the basis of $\hyb_1, \hyb_{12}$.

Consider the hybrid mesh $\calM$ composed of $n_v$ vertices $\*a_i$ and $n_{te}$
tetrahedra edges, which midpoints nodes are denoted by $\*a_{ij}$ for edge $\*a_i -
\*a_j$ (these midpoints  satisfy the geometric continuity of the
definition \ref{defn:geomesh}). We introduce two convenient notations:
\begin{itemize}
    \item[--] $\text{Sup}(\*a_i)$ is the set of \emph{cells} adjacent to $\*a_i$.
    \item[--] $\text{TSup}(\*a_{ij})$ denotes the set of \emph{tetrahedra} which contains the edge
$\*a_i - \*a_j$.
\end{itemize}
We use an intermediary function basis $((\psi_i)_i, (\psi_{ij})_{ij})$ made of
a mix of $\QQ_1, \PP_2$ finite element basis. It is defined on each elements by:
\begin{itemize}
    \item if $K \not\in \text{Sup}(\*a_i)$: $\psi_{i|K} = 0$ 
    \item if $K \in \text{Sup}(\*a_i)$ and $K$ is a hexahedron:\\
        $\psi_{i|K} = \hat \psi_i \circ \*F_K^{-1}, \ \hat \psi_i \in \QQ_1$
        with $\hat \psi_i (\hat{\*a}_j) = \delta_{ij}$ \\
        (the $\hat{\*a}_j$'s are the pre-images of $K$ vertices by $\*F_K$)
    \item if $K \in \text{Sup}(\*a_i)$ and $K$ is a tetrahedron:\\
    $\psi_{i|K} = \hat \psi_i \circ \*F_K^{-1}, \ \hat \psi_i \in \PP_2$ \\
    \strut \hfill with $\hat \psi_i (\hat{\*a}_j) = \delta_{ij}, \; \hat \psi_i (\hat{\*a}_{jk})=0$\\
        (the $\hat{\*a}_j$'s are the pre-images of $K$ vertices by $\*F_K$ and
        the $\hat{\*a}_{jk}$'s are the pre-images of the edge midpoints)
\end{itemize}
And 
\begin{itemize}
    \item if $K \not\in \text{TSup}(\*a_{ij})$: $\psi_{ij|K} = 0$ 
    \item if $K \in \text{TSup}(\*a_{ij})$:\\
    $\psi_{ij|K} = \hat \psi_{ij} \circ \*F_K^{-1}, \ \hat \psi_{ij} \in \PP_2$ \\
    \strut \hfill with $\hat \psi_{ij} (\hat{\*a}_k) = 0, \; \hat \psi_{ij} (\hat{\*a}_{kl})=\delta_{ik}\delta_{jl}$\\
        (the $\hat{\*a}_k$'s are the pre-images of $K$ vertices by $\*F_K$ and
        the $\hat{\*a}_{kl}$'s are the pre-images of the edge midpoint nodes)
\end{itemize}

It should be noticed that $\psi_i$ is discontinuous if $\*a_i$ is the vertex of
both a tetrahedron and of a hexahedron.
$\psi_{ij}$ is discontinuous if $\*a_i$ and $\*a_j$ are hexahedron vertices.
One should also notice that the $\psi_{ij}$ functions are always zeros on hexahedra.
Think of them as correcting functions to recover continuity, only defined on
tetrahedra.

To refer to previous sections, the space generated by this basis corresponds to
$\dhyb$ with constraints $(\calC\text{-}\mathrm{I}),
(\calC\text{-}\mathrm{II})$ enforced. By the adding constraint
$(\calC\text{-}\mathrm{III})$ at hybrid interfaces we ensure continuity and
form $\hyb_{12}$. By adding $(\calC\text{-}\mathrm{II}\text{-b})$ at tetrahedra
interfaces, we reduce the generated space to $\hyb_1$.

Consider the basis $(\phi_i)_{i=1..n_v}$ of $\hyb_1$ that has degrees of
freedom only on mesh vertices. We explicit its functions using linear
combination of $\psi_i, \psi_{ij}$ that depend of the local combinatorial
configuration in the hybrid mesh. We first need to introduce two more notations for
edges:
\begin{itemize}
    \item[--] $\text{ET}(\*a_i)$ is the set of tetrahedron \emph{edges}
        which contain $\*a_i$ and which are not hexahedron face diagonal (this set can be empty).
    \item[--] $\text{ETD}(\*a_i)$ is the set of tetrahedron \emph{edges} such
        that the \emph{tetrahedron} contains $\*a_i$ and the edges are
        hexahedron face diagonal. This includes edges that do not contain $\*a_i$.
\end{itemize}

We define $\phi_i$ on each element $K$ by:
\begin{itemize}
    \item if $K$ is a hexahedron: $\phi_{i|K} = \psi_{i|K}$.
    \item if $K$ is a tetrahedron: 
        \begin{align*} 
            &\phi_{i|K} = \psi_{i|K}  \\
            &+ \frac{1}{2} \sum_{\text{edge ij} \in \text{ET}(\*a_i)}
        \psi_{ij|K} + \frac{1}{4} \sum_{\text{edge jk} \in \text{ETD}(\*a_i)} \psi_{jk|K}
    \end{align*}
\end{itemize}
The last combination can be derived by continuity arguments: if $\*a_i$ is the
vertex of a hexahedron face, $\phi_i$ is equal to $1/2$ at edge $i-j$ midpoints of
the face and is equal to $1/4$ at the face center. Now consider a tetrahedron
which is connected to this face (by a triangular face containing $\*a_i$ or by
an edge containing $\*a_i$ or by a quad diagonal containing $\*a_i$ or not). At
a connecting edge $j-k$, the only function of $(\psi_i, \psi_{ij})$ in the
tetrahedron which is non-zero at $\*a_{jk}$ is $\psi_{jk}$ (which is equal to
$1$ at $\*a_{jk}$).  Applying this argument at all connecting edges of the
tetrahedron set the coefficients of the linear combination as above.

\textbf{Important remark} 
The definition of the basis $(\phi_i)_{i=1..n_v}$ is tricky because when
$\*a_i$ is the vertex of a hybrid interface, and not on the diagonal, the
function $\phi_{i|K}$ can be non-zero on a tetrahedron $K$ which does not lie in
Sup($\*a_i$). This is a consequence of the insertion of correcting functions at
interface diagonals (via $\text{ETD}(\*a_i)$ in the definition) 
even when $\*a_i$ does not lie
on the diagonal. For instance, consider the figure \ref{fig:geo_hjunction}b.),
and let $L$ be the top-left tetrahedron, $R$ be the top-right tetrahedron and
$\*a_5$ be the top vertex, then:
\begin{align*} 
    \phi_{1|L} &= \psi_{1|L} + \frac{1}{2}(\psi_{12|L} + \psi_{14|L} + \psi_{15|L})
    + \frac{1}{4} \psi_{24|L} \\
    \phi_{1|R} &= \frac{1}{4} \psi_{24|R}
\end{align*}
One should notice that $\phi_{1|R}$ is non-zero, so it has to be considered
when computing the integrals in a finite element code. This makes the assembly of
the matrices more complicated but this is required to recover the function continuity 
on a non-conforming mesh.  \QEDB \\

To form a basis of $\hyb_{12}$, one need to proceed with the same construction
except for two changes:
\begin{itemize}
    \item[--] add functions $\phi_{ij}$ defined by $\phi_{ij|K} = \psi_{ij|K}$ at tetrahedra edges $j-k$ which are
        not hexahedron edges nor hexahedron face diagonal.
    \item[--] relax the constraints: replace $\text{ET}(\*a_i)$
by $\text{ETH}(\*a_i)$, the set of tetrahedron \emph{edges} which contain
$\*a_i$ \emph{and which are hexahedron edges}, in the last point of the basis definition.
\end{itemize}

These explicit basis descriptions are useful when implementing the spaces
$\hyb_1, \hyb_{12}$ in a finite element library: one can compute the local
element contributions by combination of the standard $\QQ_1, \PP_2$ function
space contributions.

\subsection{Properties of the continuous function spaces}
The most important property of $\hyb_{12}, \hyb_1$ is that they are subspaces of
the Sobolev space $H^1$, because $H^1$ plays a fundamental role in the theory
of partial differential equations, especially for the finite element method.
Notably, it allows to apply the Lax-Milgram theorem (see \cite{ciarlet1978})
that guarantees existence and uniqueness of the solutions of the weak
formulations used in the application section.

\begin{prop}{Subspaces of Sobolev space} \\
    Let $\Omega_h$ a hybrid mesh which satisfies the definition \ref{defn:geomesh}.
    Function spaces $ \hyb_{12}, \hyb_1 $, respectively defined by
    \ref{defn:hyb12}, \ref{defn:hyb1}, are subspaces of $H_1(\Omega_h)$.

    Consequently, $ \hyb_{12}, \hyb_1 $ are Hilbert spaces.
\end{prop}
\begin{proof}
The proof can be adapted from \cite[p. 47]{ern2004}: the assumptions on the
mesh change slightly but it does not affect the rest of the proof which relies 
on the continuity of the function space.
\end{proof}


\section{Examples of applications to partial derivate equations}
In this section, we solve Poisson and the linear elasticity problems with the
continuous function spaces $\hyb_{12}, \hyb_1$. For simple problems where the
analytical solution is known, we compute errors in $L^2$-norm and compare to
standard finite elements (tri-linear hexahedra $\QQ_1$ and linear tetrahedra $\PP_1$).

\subsection{Poisson problem}
\label{sec:apppoisson}
Consider the following boundary value problem composed of the \textit{Poisson
    equation } (\ref{eq:poisson}) subject to homogeneous Dirichlet boundary conditions
(\ref{eq:bcdirichlet}).
\begin{align}
		-\Delta u=f  & \qquad \text{in} \quad \Omega \label{eq:poisson} \\
		u = 0      &  \qquad \text{on} \quad \partial \Omega \label{eq:bcdirichlet}
\end{align}
where $u$ is the unknown (temperature in heat equation for instance), $f$ a source term and
$\partial \Omega$ the domain boundary.

\paragraph*{Weak formulation}
Following the standard Galerkin approach, we project (\ref{eq:poisson}) onto a approximation
space $V$ to obtain the weak formulation (\ref{eq:poissonweak}).
\begin{align}
    \label{eq:poissonweak}
    \forall v \in V, \; \int_\Omega \! - \Delta u \ v \ \mathrm{d}x &= \int_\Omega \!  f \ v \ \mathrm{d}x
\end{align}
Since $V$ is sufficiently regular, namely $V \subset C^0(\Omega) \cap
H^1(\Omega)$ (this is the case for $\hyb_1, \hyb_{12}$), we can use
integration by parts formula. The Dirichlet boundary condition is taken into account by restricting
ourselves to $V_0 = \{ v \in V \text{ such that } v = 0 \text{ on } \partial \Omega \}$. Thus the
weak formulation becomes:
\begin{equation} \label{eq:poissonweakd}
    \forall v \in V_0, \ \int_\Omega \!  \nabla u \ \nabla v \ \mathrm{d}x = \int_{\Omega} \! f \ v \ \mathrm{d}x  
\end{equation}
(see finite element textbooks such as \cite{ciarlet1978}, \cite{allaire2007},
\cite{ern2004} for detailed derivations and proofs)

\paragraph*{Finite element discretization} 
We now consider the finite dimension subspace $V_h \subset V$, $V_h = \hyb_1 \text{ or } V_h = \hyb_{12}$.
Functions $u$ and $v$ can be both decomposed onto the $(\phi_i)_{i=1..n}$ function basis of $V_h$. Then
(\ref{eq:poissonweakd}) becomes a linear system of equations $A x = B$ with:

$$ 
A_{ij} = \int_\Omega \!  \nabla \phi_i \nabla \phi_j \mathrm{d}x \quad
\text{and} \quad B_i = \int_{\Omega} \! f \phi_i \mathrm{d}x 
$$
These integrals are decomposed over elements. For each element, a change of
variable is used to get back to the reference element (and the chain rule if
derivatives are involved). For instance, the contribution of element $K_c$ to
the coefficient $B_i$ is: $$ B_{i|K_c} = \int_{\hat K_c} \!
f(\*F_{K_c}(\hat{\*x})) \ \hat{\phi}_{i}(\hat{\*x}) \;
|\text{det}(J_{\*F_{K_c}}(\hat{\*x}))| \mathrm{d}\hat{\*x}  $$

The integrals are computed using numerical quadrature, i.e. evaluating
operand at well-chosen locations, so values taken by $\hat \phi_i, \nabla \hat
\phi_i$ can be pre-computed on reference elements and used for computations
on actual elements. For each element, one also needs to compute the Jacobian of
the mappings at quadrature points (which is not constant for
tri-affine and quadratic mappings).

\paragraph*{Numerical validation on analytical Poisson problem}
Consider the simple following Poisson problem on the unit cube $\Omega = [0, 1]^3$:
\begin{align} 
    -\Delta u =& 3 \pi^2 sin(\pi x)*sin(\pi y) * sin(\pi z) 
    &&\text{in} \quad \Omega \label{eq:sinbump}\\ 
    u =& \ 0    &&\text{on} \quad \partial \Omega_D  \nonumber 
\end{align}
Its analytical solution is $u = \sin(\pi x)*\sin(\pi y) * \sin(\pi z)$.

\paragraph*{Numerical experiment setup}

The meshes that we use are built using the following procedure: (a) the unit cube
is regularly divided in smaller cubes, (b) vertices inside the cube are
randomly displaced within a range up to d\% of mean edge length, (c) some cubes
are transformed into 6 tetrahedra. Distortion of the mesh (also used in
\cite{bergot2010}) applied in step (b) ensures elements are not parallel to
borders and that hexahedra faces are not planar. This is an attempt to
eliminate specific artefacts associated with unrealistic regularity of the
mesh. Transformation of hexahedra into tetrahedra (step (c)) is used to
generate hexahedral-tetrahedral meshes or fully tetrahedral meshes.

For numerical experiments, 20\% of hexahedra are transformed in tetrahedra,
resulting in a hybrid mesh where tetrahedra are 60\% of overall elements. This
proportion is largely superior to typical outputs of hex-dominant meshing
algorithms. The distortion of interior vertices is set to $d = 10\%$, this
produces dihedral angles with an average of 9 degrees and a maximum at 42 degrees
for quadrilateral faces if we consider them as two triangles (see figure
\ref{fig:geo_hjunction}a.). From our experience with hex-dominant meshes
\cite{sokolov2015}, they are typical non-planarity angles. An example of hybrid
mesh built with this procedure is shown in figure \ref{fig:cube_hexdom}.

\begin{figure}
  \centering
  \includegraphics[width=0.8\columnwidth]{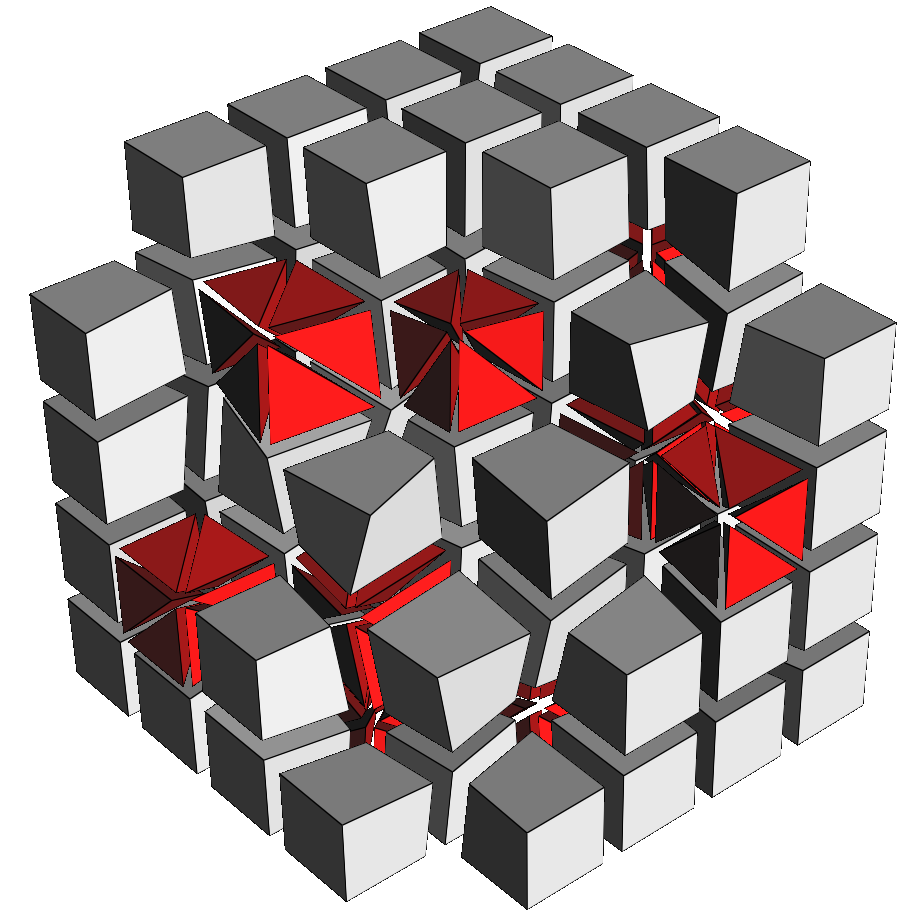} 
  \caption{Example of hybrid hexahedral-tetrahedral mesh of the unit cube with distortion $d=10\%$.
        The tetrahedra are colored in red and the hexahedra in grey.}
  \label{fig:cube_hexdom}
\end{figure}

\paragraph*{Results}
The function spaces $\hyb_1, \hyb_{12}$ have been implemented in a modified
version of the open source library MFEM \cite{mfem-library}. We solve the
analytical problem (\ref{eq:sinbump}) with finite element basis $\hyb_1,
\hyb_{12}, \PP_1, \QQ_1$ on meshes successively refined. Relative errors in
$L^2$-norm are reported in figure \ref{fig:sinbump}. In x-axis, we use $(\#
\text{degree of freedom})^{\frac{1}{3}}$ which is proportional to inverse of
the cell sizes in our cubic configuration.

\begin{figure}
  \centering
  \includegraphics[width=\columnwidth]{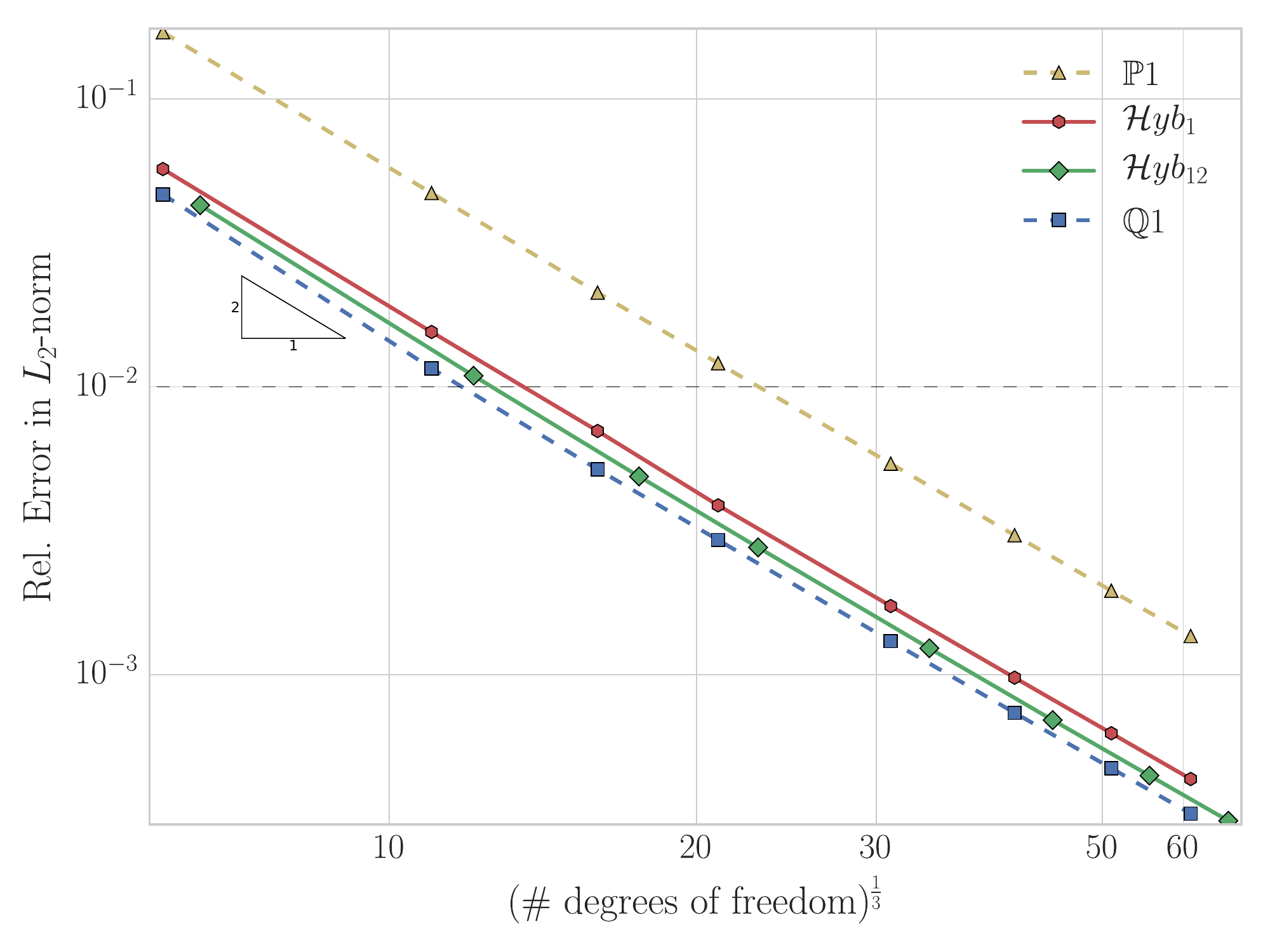} 
  \caption{Finite element simulation errors on Poisson analytical problem for
      different finite element spaces. $\hyb_1$ and  $\hyb_{12}$ solutions are
      close to $\QQ_1$ solutions and are significantly better than $\PP_1$
      solutions. Both axis use logarithmic scale.}
  \label{fig:sinbump}
\end{figure}

We observe that error convergence rates are quadratic in $L^2$-norm with mesh
refinement. For $\PP_1, \QQ_1$, this is the optimal convergence rate, see
\cite{ciarlet1978}.  For $\hyb_1, \hyb_{12}$, this could be expected as they
are made of $\QQ_1$ and $\PP_2$ with added linear constraints. The interesting
part is hybrid function spaces are much more closer to $\QQ_1$ than to $\PP_1$.
Measured accuracy with $\hyb_1$ is three times better than with $\PP_1$. Thus
solutions computed with the introduced spaces achieve good accuracy, 1\% for
instance, with much less refined meshes, and consequently smaller linear
systems.
Figure \ref{fig:sinbump_time} shows the same computations with the elapsed
times in x-axis. These timings include the assembly of the linear system and
the solve time of the iterative conjugate gradient solver (reduction of the
residual by a factor $10^{10}$). These results still indicate a gain of 
$\hyb_1, \hyb_{12}$ spaces over $\PP_1$. The timings obtained for small meshes
(time $< 0.5$ seconds) do not carry useful information as they are too much
influenced by external parameters such as processor cache, other jobs running,
etc.  It should also be reported that our implementation can be significantly
improved as it currently uses a large linear system (corresponding to the
non-constrained space) which is then reduced by applying constraints as
matrix-matrix multiplications. This can be avoided by computing directly the
right linear system to reduce execution-times, as suggested in section
\ref{sec:basis}.

\begin{figure}[h]
  \centering
  \includegraphics[width=\columnwidth]{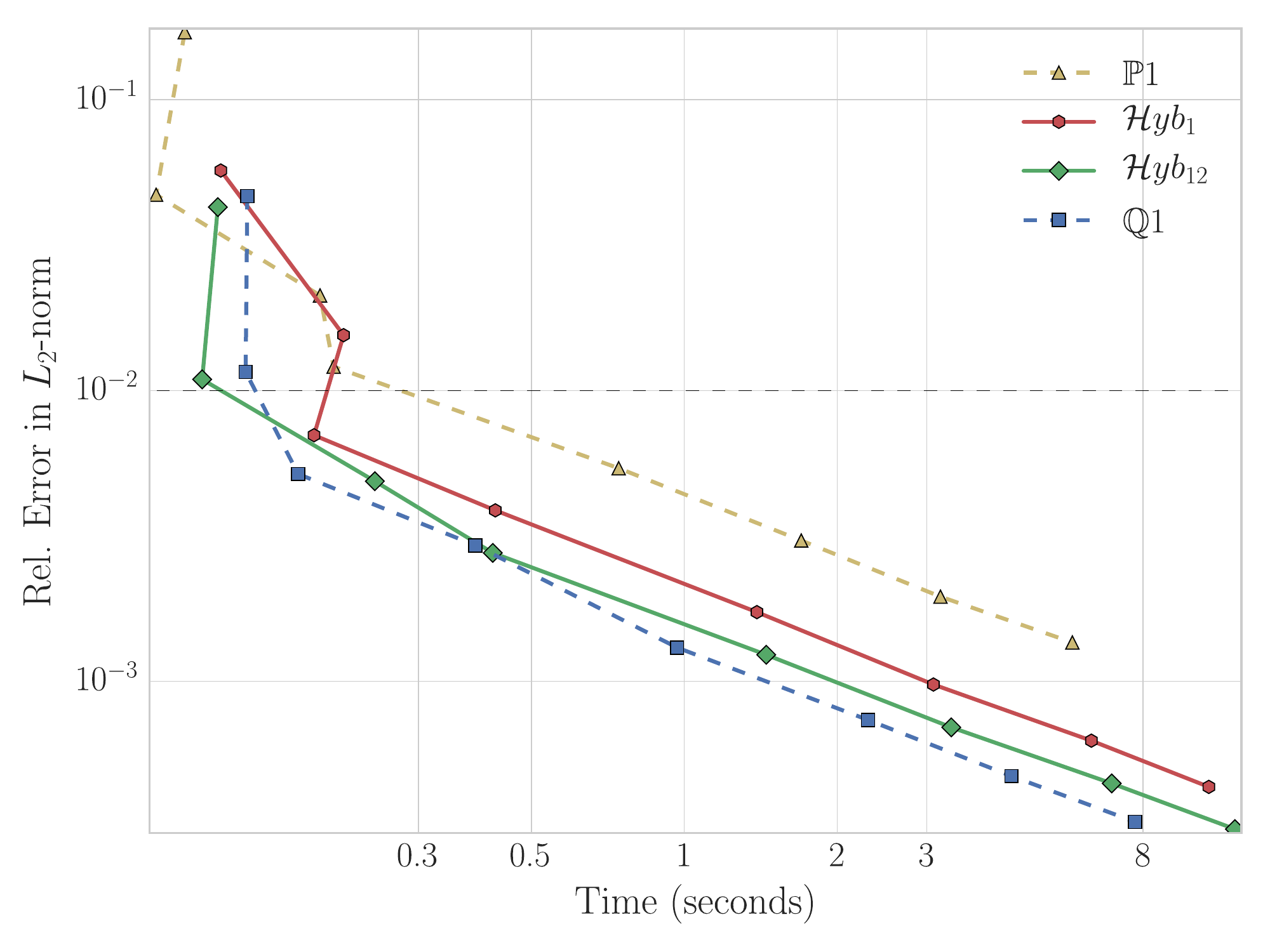} 
  \caption{Relative errors in $L^2$-norm for various finite element spaces as functions
      of computing time (assembly + solver). Both axis use logarithmic scale.}
  \label{fig:sinbump_time}
\end{figure}

\paragraph*{Importance of tetrahedra quadratic mappings}
Quadratic mappings for tetrahedra at non-conforming interfaces (introduced in
proposition \ref{prop:hjunction} to recover the continuity of the geometry) can
be seen as superfluous in the context of finite element simulations as there are
other sources of numerical errors. To highlight their impact, we solve the same
analytical problem with $\hyb_1$ using affine and quadratic mappings for
distortion values $d=10\%$ and $d=20\%$. The results are reported in figure
\ref{fig:sinbump_geoP2}.

This experiment shows that for high accuracy (error $< 3\%$), the use of affine
mappings instead of quadratic ones for tetrahedra of hybrid junctions can
induce significant errors. As one can expect, this error is tightly linked to
the degree of non-planarity of hexahedron faces. So our advice is to check the
quality of hexahedron faces in a pre-processing phase, and if the quality is
high (typically dihedral angle of quadrilateral faces $< 5$ degrees), the usage of affine
mapping approximation is reasonable unless high accuracy is desired. An
experiment with $\hyb_{12}$ exhibits exactly the same behavior (loss of
convergence when using affine mappings).

\begin{figure}[!h]
  \centering
  \includegraphics[width=\columnwidth]{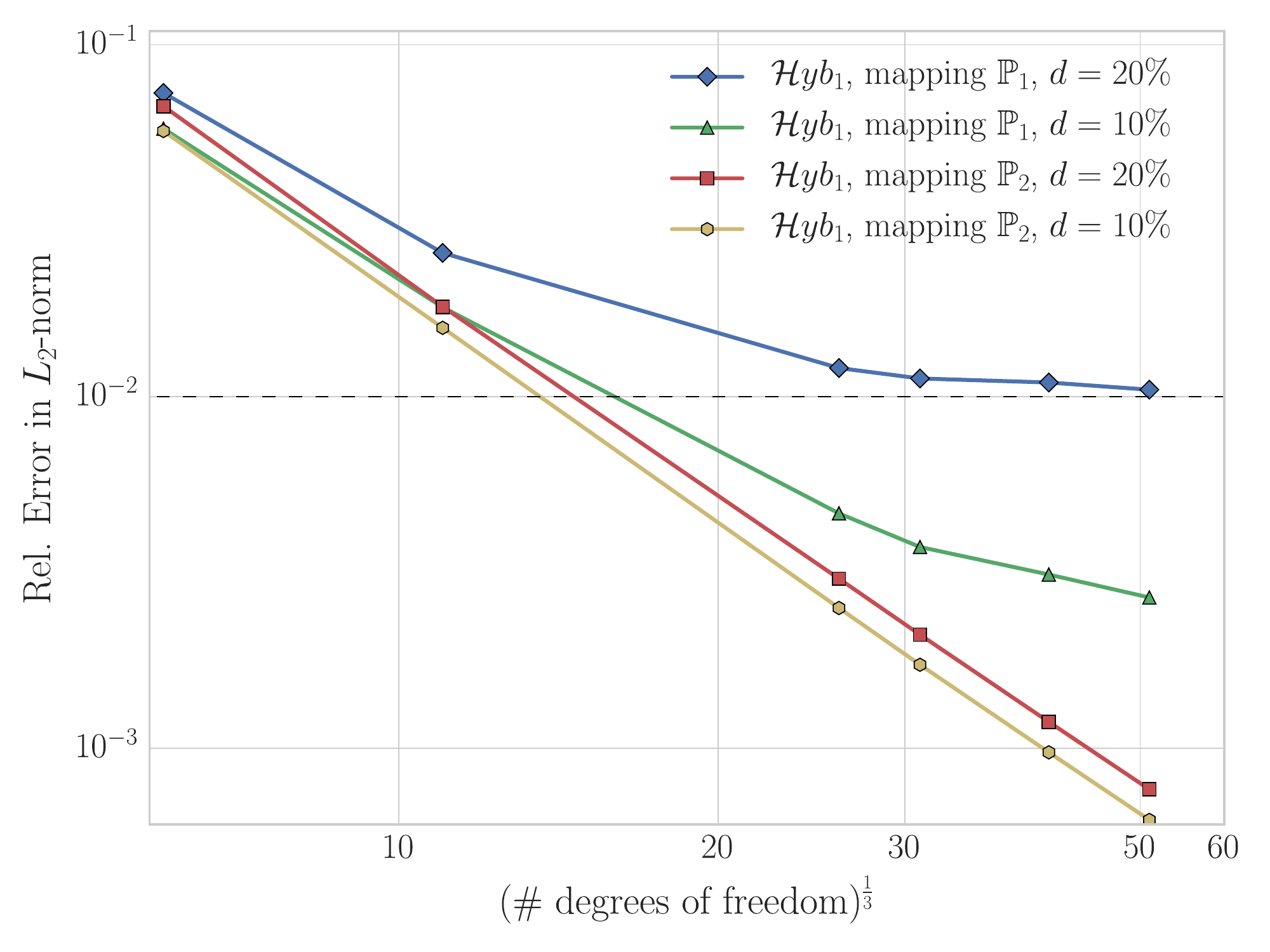} 
  \caption{Influence of tetrahedra mappings at non-conforming interfaces on the
      analytical Poisson problem, for 10\% and 20\% edge length displacement of
      vertices in the cube. Affine mappings $\PP_1$ and quadratic mappings
      $\PP_2$.}
  \label{fig:sinbump_geoP2}
\end{figure}

\subsection{Linear elasticity}
The system of equations of linear elasticity is the usual description for continuum
mechanics with small deformations. Consider a deformable medium $\Omega$ at
equilibrium, fixed on $\partial \Omega_D$, subject to a volumetric load $\*f$
inside $\Omega$ and to a surface force $\*g$ on the boundary $\partial \Omega_N$. The material
behavior is given by the Hooke's law (Lam\'e parameters $\lambda, \mu$). The resulting
displacement field $\*u \in \RR^3$ is governed by the system:
\begin{align}
    \nabla \cdot \sigma(\*u) + \*f = 0  & \qquad \text{in} \quad \Omega  \\
    \sigma (\*u) = \lambda(\nabla \cdot \*u) \mathcal{I} 
        + \mu (\nabla \cdot \*u + \nabla \cdot \*u^T )  & \qquad \text{in} \quad \Omega  \\
    \*u = 0      &  \qquad \text{on} \quad \partial \Omega_D \label{eq:bcdirichlete}\\
	\sigma (\*u) \cdot \*n = g & \qquad \text{on} \quad \partial \Omega_N \label{eq:bcneumanne}
\end{align}
where $\*n$ is the exterior normal and $\mathcal{I}$ the identity matrix.

\paragraph*{Weak formulation} For the weak formulation of the elasticity problem,
we consider the simple displacement formulation (\ref{eq:elasticityweak}). The
derivation is similar to the Poisson problem but longer, the reader can refers to
\cite{ciarlet1978}, \cite{ern2004} or other textbooks for the details.
\begin{align}
    \label{eq:elasticityweak}
    \forall \*v \in (V_0)^3, \ \int_\Omega \!  \nabla \cdot \*u \ \nabla \cdot \*v 
    \ + & \ 2 \ \mu \ \epsilon (\*u) : \epsilon(\*v) \ \mathrm{d}x \\
    &= \int_{\Omega} \! \*f \cdot \*v \ \mathrm{d}x  
    +\int_{\partial \Omega_N} \! \*g \cdot \*v \ \mathrm{d}s \nonumber
\end{align}
where $\epsilon (\*u) = \frac{1}{2}(\nabla \cdot \*u + \nabla \cdot \*u^T)$.

\paragraph*{Validation on analytical linear elasticity problem}
The following experiment solves the static linear elasticity problem
with homogeneous Dirichlet boundary conditions and a load
applied inside the domain $\Omega = [0, 1]^3$. The problem is borrowed from \cite{schillinger2015}.
\begin{align} 
    \nabla \cdot \sigma(\*u) + \*f = 0  & \qquad \text{in} \quad \Omega  \nonumber \\
    \sigma (\*u) = \lambda(\nabla \cdot \*u) \mathcal{I} 
        + \mu (\nabla \cdot \*u + \nabla \cdot \*u^T )  & \qquad \text{in} \quad \Omega  \label{eq:lesin} \\
    \*u = 0      &  \qquad \text{on} \quad \partial \Omega \nonumber 
\end{align}
where the loading $\*f$ and the Lam\'e parameters are detailed in \cite{schillinger2015}. 
The analytical expression of the
displacement is $$u_x = u_y = u_z = \sin(2 \pi x) \sin(2 \pi y) \sin(2 \pi x)$$

The procedure for mesh generation is exactly the same as with the
analytical Poisson problem experiment. Relative errors in $L^2$-norm are shown
in figure \ref{fig:lesin}. The conclusions drawn with the analytical Poisson
problem apply here too: solutions computed with the hybrid space $\hyb_1,
\hyb_{12}$ are close to the tri-linear ones ($\QQ_1$) and significantly more
accurate than solutions obtained with tetrahedra linear elements $\PP_1$. 
On this specific example, the $\hyb_1$ solution is $3.5$ times more accurate than
the $\PP_1$ one with the same number of degrees of freedom.

\begin{figure}
  \centering
  \includegraphics[width=\columnwidth]{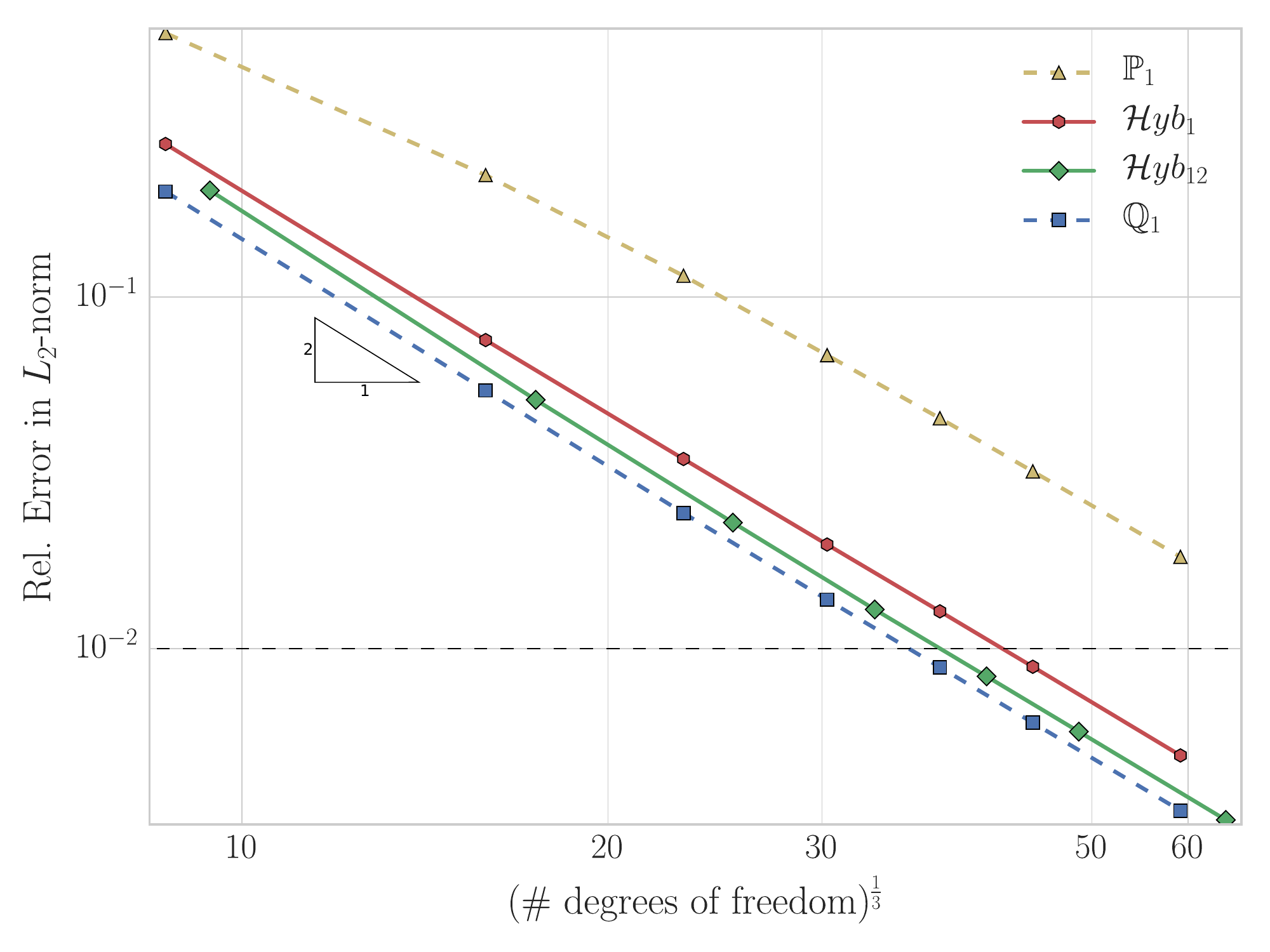} 
  \caption{Finite element simulation error on linear elasticity analytical
      problem for different finite element spaces. $\hyb_1$ and  $\hyb_{12}$
      solutions are close to $\QQ_1$ solutions and are significantly better
      than $\PP_1$ solutions.}
  \label{fig:lesin}
\end{figure}

\paragraph*{Simulations on more complex meshes}
Besides the standard test cases, we applied our approach to hex-dominant meshes
generated from industrial 3D models. The figure \ref{fig:hanger} illustrates a
linear elasticity problem solved with $\hyb_1$ on a hexahedral-tetrahedra mesh
generated with \cite{sokolov2015}. The 3D model \textit{hanger} is borrowed from
\cite{livesu2015}.  The solutions computed are consistent with the ones
computed with standard Lagrange basis but further work is required to quantify
precisely the differences. Indeed there are no analytical solution for
non-trivial geometries and computing accurately a distance between finite
element solutions defined on distinct meshes is not straightforward.

\begin{figure*}
  \centering
  \includegraphics[width=\textwidth]{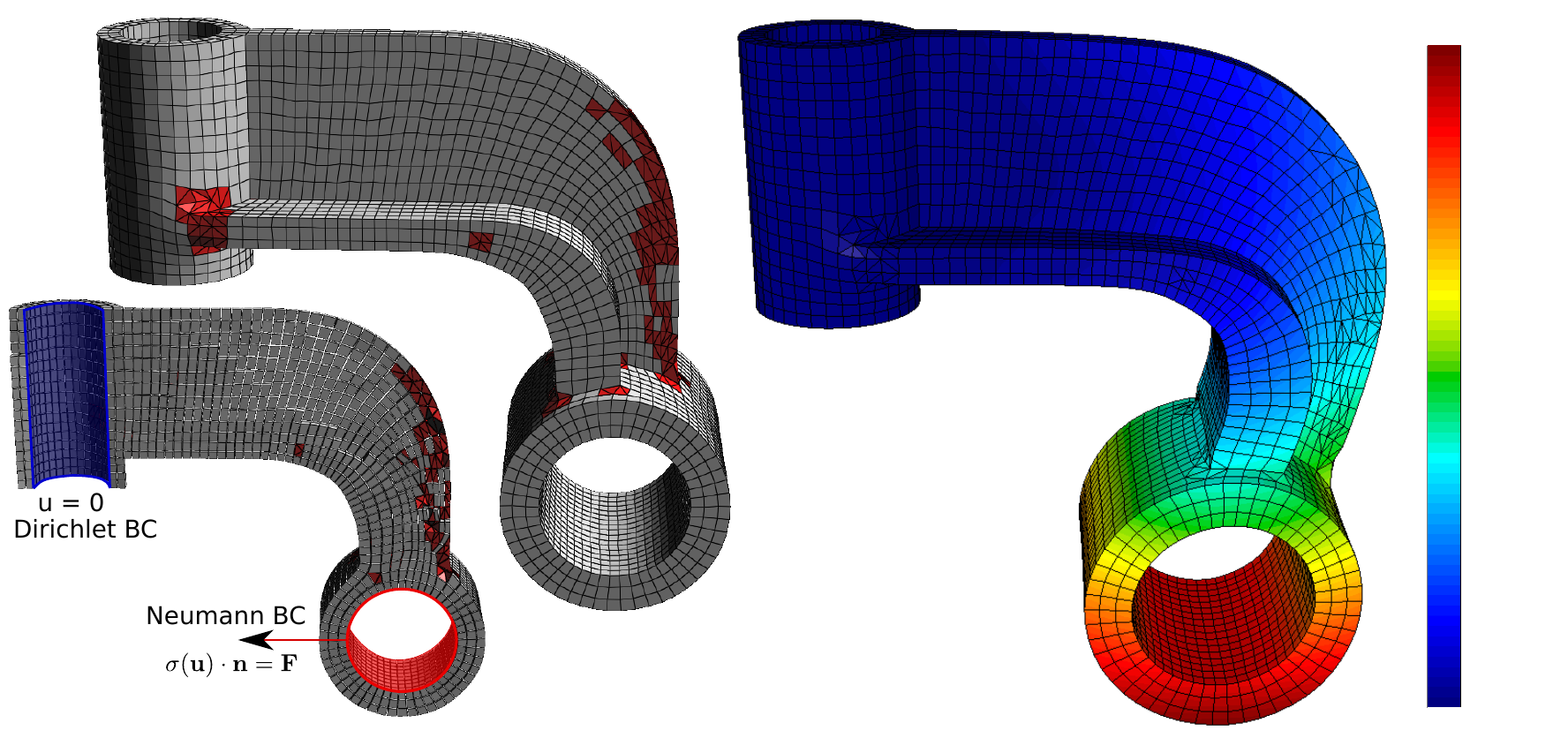} 
  \caption{$\hyb_1$-finite element solution on hybrid hexahedral-tetrahedral mesh. On the left figure,
  tetrahedra are colored in red and hexahedra in grey. On the right figure, the color is the magnitude 
  of the displacement field.}
  \label{fig:hanger}
\end{figure*}


\section{Conclusion}

Two continuous function spaces, $\hyb_1$ and $\hyb_{12}$, defined on hybrid
hexahedral-tetrahedral meshes have been introduced. The continuity of both the
geometry and the function spaces is recovered using quadratic mappings and
quadratic functions on tetrahedra connected to hexahedra, with
constraints at hybrid junctions.

The experiments conducted on analytical problems with smooth solutions show
that $\hyb_1, \hyb_{12}$ defined on hybrid meshes perform better (factor 3 in
our tests) than $\PP_1$ (tetrahedral meshes) and slightly worst than $\QQ_1$
(hexahedral meshes). We conjecture with confidence that, under standard mesh
shape and function regularity assumptions, $\hyb_1$ and $\hyb_{12}$ have a
quadratic convergence rate in $L^2$-norm with mesh refinement and a linear
convergence rate in $H_1$-norm.

Even if our current implementation works on any geometry, further research is
required to quantify the errors obtained when applying the method on
non-trivial geometries, especially the impact of hexahedral-tetrahedral meshes
properties (proportion of tetrahedra, quality of elements). This requires
techniques to compute distance between finite element solutions computed on
distinct meshes. We are currently working on this topic.

Possible future work can be the extension of the proposed function spaces to
higher orders by using standard Lagrange finite elements $\QQ_k, \PP_k$. For
hexahedra $\QQ_k$, functions restricted to faces are bi-variate polynomials of degree
$2 k$, so it should be possible to build \emph{continuous} function spaces of order
$k$ with a mix of $\QQ_k$ and $\PP_{2k}$ finite elements subjected to
appropriate constraints at hybrid interfaces.

\appendix
\label{appendix}
\section{Function basis of standard Lagrange finite elements}
\subsection{Reference tetrahedron and $\PP_1, \PP_2$ function spaces} \label{app:tet}
The reference tetrahedron, denoted by $\hat T$, is defined by its 4 vertices 
$\hat{\*s}_1 = (0, 0, 0),\ \hat{\*s}_2 = (1, 0, 0),\ \hat{\*s}_3 = (0, 1, 0),
\ \hat{\*s}_4 = (0, 0, 1)$.
The barycentric coordinates of $\hat T$ are:
\begin{align*}
    \hat \lambda_1(u, v, w) &= 1-u-v-w, \ &&\hat \lambda_2(u,v,w) = u  \\
    \hat \lambda_3(u, v, w) &= v,  \ &&\hat \lambda_4(u, v, w) = w  
\end{align*}
They satisfy $\forall i \in [1, 4], \; \lambda_i(\hat{\*s}_j) = \delta_{ij}$, so
they form a basis of the space of tri-variate polynomials of degree 1:
$$ \PP_1 = \{p(x, y, z) = a x + b y + c z + d \text{ with } a, b, c, d \in \RR\} $$
The decomposition on the basis is:
$$ \forall p \in \PP_1, \; p(u, v, w) = \sum_{i=1}^4 p_i \ \hl_i (u, v, w) \; 
\text{ where } p_i = p(\hat{\*s}_i)$$

The space of tri-variate polynomials of degree 2 is:
$$ \PP_2 = \{p(x, y, z) = \sum_{0 \leq i+j+k \leq 2} a_{ijk} x^i y^j z^k  \text{ with } a_{ijk} \in \RR\} $$
Its interpolating basis $(\hat \phi_i)_{i=1..10}$ can be expressed in terms of the barycentric coordinates:
\begin{align*}
    \hphi_i(u, v, w) &= \hl_i (2 \hl_i -1) &&1 \leq i \leq 4 \\
    \hphi_{ij}(u, v, w) &= 4 \hl_i \hl_j      &&1 \leq i < j \leq 4 
\end{align*}
The first four functions are associated with the vertices $\hat{\*s}_i$ of $\hat T$ and
the last six functions are associated with the edge midpoints $\hat{\*s}_{ij} =
\frac{\hat{\*s}_i + \hat{\*s}_j}{2}$. The decomposition is: 
$$ \forall p \in \PP_2, \; p(x, y, w) = \sum_{i=1}^{4} p_i \hphi_i (x, y, z) +
\sum_{1 \leq i < j \leq 4} p_{ij} \hphi_{ij}(x, y, z) $$
where $p_i = p(\hat{\*s_i}), p_{ij} = p(\hat{\*s}_{ij})$

\begin{nprop} \label{prop:restrictionP2}
    The restriction $p_{|t}$ of $p \in \PP_2$ to a face $t \subset \hat T$ is a 
    bi-variate polynomial of degree 2, which has 6 coefficients determined by
    the values of $p_{|t}$ at the 6 points $\hat{\*s}_i, \hat{\*s}_{ij} \in t$.
\end{nprop}

\subsection{Reference hexahedron and $\QQ_1$ function space} \label{app:cube}
In this work, the reference hexahedron $\hat Q$ is the unit cube $[0, 1][0, 1][0, 1]$. 
The difference with the tetrahedron is that there are no barycentric coordinates but
there is a symmetry of the cell along the three axis that we can exploit. We denote
$(\hat{\*q}_i)_{i=1..8}$ the vertices of $\hat Q$ : 
$\hat{\*q}_1 = (0, 0, 0), \ \hat{\*q}_2 = (1, 0, 0), \ \hat{\*q}_3 = (1, 1, 0),\ \hat{\*q}_4 = (0, 1, 0),\ \text{etc.}$

By product of degree one polynomials $(x_i), (1 - x_i)$ defined along each axis, 
we can build the set $(\hps_i)_{i=1..8}$ as follow:
\begin{align*}
    \hps_1 &= (1-u)(1-v)(1-w)  && \hps_5 = (1-u)(1-v)w  \\
    \hps_2 &= u(1-v)(1-w)      && \hps_6 = u(1-v)w      \\
    \hps_3 &= u v (1-w)        && \hps_7 = u v w        \\
    \hps_4 &= (u-1) v (1-w)    && \hps_8 = (u-1) v w    
\end{align*}
They satisfy $ \hps_i(\hat{\*q}_j) = \delta_{ij}, \ 1 \leq i, j \leq 8 $ and form
a basis of the space of tri-variate polynomials of degree one in each variable.
$$ \QQ_1 = \{p(x, y, z) = \sum_{0 \leq i, j, k \leq 1} a_{ijk} x^i y^j z^k  \text{ with } a_{ijk} \in \RR\} $$
We have the decomposition 
$$ \forall p \in \QQ_1, \ p(x, y, w) = \sum_{i=1}^{8} p_i \hps_i (x, y, z) \; \text{ where } p_i = p(\hat{\*q}_i)$$
These polynomials are said to be tri-affine.

\begin{nprop} \label{prop:q1face}
    The restriction $p_{|q}$ of $p \in \QQ_1$ to a face $q \subset \hat Q$ is a bi-variate
    polynomial of degree 1 in each variable, which has 4 coefficients
    determined by values of $p_{|q}$ at the 4 vertices of the face $q$.
\end{nprop}

\bibliographystyle{elsarticle-num} 
\bibliography{hybrid}

\end{document}